\documentclass[conference]{IEEEtran}
\IEEEoverridecommandlockouts

\usepackage{textcomp}
\usepackage{xcolor}

\usepackage{graphicx}
\usepackage{color}
\usepackage{rotating}
\usepackage{tabularx}
\usepackage{pdflscape}
\usepackage{amsmath,amsthm,amssymb,bm}
\usepackage{blkarray}
\usepackage{algorithm,algorithmic}
\usepackage{bbm,dsfont}
\usepackage{caption,subcaption}
\usepackage{wrapfig}
\usepackage{cancel}
\usepackage{enumerate, cases}
\usepackage{thmtools,thm-restate}
\usepackage{mathtools}
\usepackage{cite}

\usepackage[none]{hyphenat}
\usepackage{multirow}
\usepackage{makecell}
\usepackage{setspace}
\usepackage{pifont}

\usepackage{dblfloatfix}
\usepackage{array}
\usepackage{enumitem}
\usepackage{tikz}
\usetikzlibrary{decorations.pathreplacing}
\usepackage{pgfkeys}
\usetikzlibrary{intersections}

\DeclarePairedDelimiter\abs{\lvert}{\rvert}%
\makeatletter
\let\oldabs\abs
\def\abs{\@ifstar{\oldabs}{\oldabs*}}

\DeclareMathOperator*{\argmax}{arg\,max}

\newcolumntype{C}[1]{>{\centering}m{#1}}
\newcommand{\EE}[1]{\mathbb{E}\left[#1\right]}

\newcommand{\ceil}[1]{\left\lceil #1 \right\rceil}

\newcommand{\Algo}{{\it Holographic Beam (HoloBeam) }}
\newcommand{\HB}{HoloBeam}
\newcommand{\AlgoI}{$\beta_1$-HoloBeam}
\newcommand{\AlgoII}{$\beta_2$-HoloBeam}
\newcommand{\Algoi}{$\beta_i$-HoloBeam}
\newcommand{\Pro}{\Pr}

\theoremstyle{plain}
\newtheorem{thm}{Theorem}
\newtheorem{lem}{Lemma}
\newtheorem{prop}{Proposition}

\newtheorem{defi}{Definition}

\def\BibTeX{{\rm B\kern-.05em{\sc i\kern-.025em b}\kern-.08em
    T\kern-.1667em\lower.7ex\hbox{E}\kern-.125emX}}
\begin{document}

\title{HoloBeam: Learning Optimal Beamforming in Far-Field Holographic Metasurface Transceivers\\
\thanks{
Manjesh K. Hanawal thanks funding support from SERB, Govt of India, through Core Research Grant (CRG/2022/008807) and MATRICS grant (MTR/2021/000645), and DST-Inria Targeted Programme.}}


\author{Debamita Ghosh, IITB-Monash Research Academy, India,  e-mail: debamita.ghosh@iitb.ac.in  \\ Manjesh K. Hanawal, MLiONS Lab, IEOR, IIT Bombay, India, e-mail: mhanawal@iitb.ac.in \\ Nikola Zlatanov, Innopolis University, Russia, e-mail: n.zlatanov@innopolis.ru }
\maketitle

\begin{abstract}
 Holographic Metasurface Transceivers (HMTs) are emerging as cost-effective substitutes to large antenna arrays for beamforming in Millimeter and TeraHertz wave communication. However, to achieve desired channel gains through beamforming in HMT, phase-shifts of a large number of elements need to be appropriately set, which is challenging. Also, these optimal phase-shifts depend on the location of the receivers, which could be unknown. In this work, we develop a learning algorithm using a {\it fixed-budget multi-armed bandit framework} to beamform and maximize received signal strength at the receiver for far-field regions. Our algorithm, named \Algo exploits the parametric form of channel gains of the beams, which can be expressed in terms of two {\it phase-shifting parameters}. Even after parameterization, the problem is still challenging as phase-shifting parameters take continuous values. To overcome this, {\it\HB} works with the discrete values of phase-shifting parameters and exploits their unimodal relations with channel gains to learn the optimal values faster. We upper bound the probability of {\it\HB} incorrectly identifying the (discrete) optimal phase-shift parameters in terms of the number of pilots used in learning. We show that this probability decays exponentially with the number of pilot signals. We demonstrate that {\it\HB} outperforms state-of-the-art algorithms through extensive simulations.

\end{abstract}


\section{Introduction}
\label{sec:intro}
The demand for high-data rates is growing with the emergence of data-intensive applications, which are expected to offer energy efficiency and low latency. 
New solutions are needed to enhance the network capacity and connectivity in the next-generation wireless networks to support these applications. Millimeter wave (mmWave) and TeraHertz (THz) communication technologies are potential solutions due to the availability of large bandwidths at very high frequencies. However, mmWave and THz signals are 
prone to severe deterioration from reflections and absorption, significantly restricting their direct adaptation in cellular networks \cite{wan2021terahertz}, \cite{jamali2020intelligent}. 

Base stations (BSs) with massive antenna arrays can be used to focus the transmission energy in the desired direction through beamforming. This can offer high beamforming 
gains and compensate for the signal degradation \cite{larsson2014massive}. However, implementing a BS with a large antenna array increases hardware complexity and costs. HMTs are emerging as a low-cost solution for building a massive antenna array \cite{wu2019intelligent, huang2020holographic, huang2018metasurface}. 

An HMT is a rectangular surface consisting of nearly infinite low-cost metamaterial elements densely deployed into a limited surface area to form a spatially continuous transceiver aperture. HMTs can be considered as an extension of the traditional massive antenna arrays with discrete antennas to the continuous reflecting surfaces \cite{hu2018beyond}. Each metamaterial element in HMT acts as a phase-shifting antenna that can change the phase of transmitting/receiving signals and thereby  beamform the transmit/receive signals to/from the desired direction with high channel gains and enhancing the received signal strength (RSS) at the receiver \cite{hu2018beyond}. 
The channel gain from HMT to receivers depends on the phase-shift at each element. Identifying the correct phase-shifts for all the elements is a fundamental and challenging task \cite{yoo2021holographic, zhang2022beam, demir2022channel}, since the required phase-shifts depend on various factors such as channel state information (CSI) and user location, which are unknown. Our goal in this work is to develop a beamforming algorithm that learns optimal phase-shifts and maximizes RSS in HMT systems. As the number of phase-shifting elements in an HMT can be large, learning the optimal phase-shifts is not trivial. We tackle this task by exploiting the structural properties of the channel gains in the phase-shift values. 


In the far-field region, the channel gain of beams can be expressed in a parametric form by setting each phase-shift of the HMT as a function of two common {\it phase-shifting parameters} \cite{ghermezcheshmeh2021channel, Arxiv2022ghoshub3}. 
To maximize the RSS at the receiver location, these parameters must be set to a value that corresponds to the receiver's location. However, in practice, the receiver's location is unknown,  and the phase-shifting parameters associated with its location must be learned. In particular, these parameters are related to the azimuth and elevation of the receiver's location from the center of the HMT and need to be estimated by measuring the RSS received at the receiver. We exploit this fact to develop a learning algorithm to beamform HMT towards the receiver and maximize RSS. 

The two phase-shifting parameters take continuous values, and learning the optimal values is not practically feasible in a finite time. Therefore, we discretize the range of these parameters and search for the optimal values over the discretized space. The discretization is done in such a way that optimal values in the discrete space and the continuous space are close. 
We use a fixed-budget multi-armed bandit (MAB) framework to develop an algorithm named \Algo to identify the optimal phase-shifting parameters in the discrete space. The algorithm works in two phases. Each phase learns the optimal value of one parameter while keeping the other parameter constant. Further, the algorithm exploits the {\it unimodal structure} of the  RSS for each parameter to reduce the pilots required to identify the optimal value of these parameters. The algorithm works in batches where 
one-third of discrete parameter values are eliminated at the end of every batch.   

Next, we provide theoretical guarantees on the performance of {\it \HB}. 
The analysis uses tail bounds of the non-central Chi-squared distribution to demonstrate that the error probability goes to zero exponentially in the number of pilot symbols used to estimate the parameters. In particular, we establish that {\it \HB}  achieves an error probability of the order of $\mathcal{O}\bigg(\sum\limits_{i=1}^2\log_2K_i\exp\left\{-\frac{n\Delta_i^2}{\sigma^4(1+3G/\sigma^2)^4}\right\}\bigg)$ after using $n$ pilot symbols during channel estimation, where $K_i$ denote the number of phase-shifting parameters in the discrete space, $\Delta_i$ is the minimum gap between mean RSS of successive values in the discrete set of $\beta_i$ for $i=1,2$, $\sigma^2$ is the variance of noise, and $G$ is associated with dimensions and radiation properties of the HMT.

In summary, our main contributions are as follows:
\begin{itemize}
    \item We exploit the parametric properties of the far-filed channel gains to 
    set up the beamforming problem in the HMT system as a fixed-budget MAB problem. 
\item We exploit the unimodal structure of the far-field channel gains and propose an algorithm named {\it \HB} that identifies the best phase-shift parameters. 
\item We upper bound the error probability of {\it \HB} and demonstrate that it is strictly smaller than the case when no unimodality is exploited. 
\item Extensive numerical simulations validated that {\it \HB } significantly outperforms the state-of-the-art algorithms.


\end{itemize}

\section{Related Work}
\label{sec:relwork}
Some prominent channel estimation schemes such as exhaustive search \cite{dai2006efficient}, hierarchical search \cite{xiao2016hierarchical, chen2018beam, zhang2017codebook}, and compressed sensing \cite{chen2018beam, zhang2017codebook, CSpeng2015enhanced, CStsai2018efficient} are applicable to the HMT systems. However, all these schemes require high training overheads and system latency. Some CSI estimation  schemes, such as least-square estimation-based approach \cite{WCL2022channelestimationHMMIMO}, and subspace-based approach \cite{demir2022channel} are developed for holographic MIMO (HMIMO) communications. However, the computational complexity scales up with the number of phase-shift elements in \cite{WCL2022channelestimationHMMIMO} and \cite{demir2022channel}. The authors in \cite{an2023tutorial} give an overview of the efficient channel estimation approaches on HMIMO communications.
An efficient CSI scheme of line-of-sight (LOS) dominated far-field channel between an HMT and a user is proposed in \cite{ghermezcheshmeh2021channel}. The computational cost and the training overhead of the proposed scheme do not scale with the number of phase-shift elements of the HMT, but they did not provide any theoretical guarantee on their proposed algorithm. A pure-exploration based algorithm with a theoretical guarantee is proposed in \cite{Arxiv2022learningHMT}, which outperforms the one in \cite{ghermezcheshmeh2021channel}. However, none of the above-mentioned works exploit the unimodal structure of the far-field channel model of the HMT system.

The papers \cite{INFOCOM2018_EfficientBeamAllignment, Infocom2020_MAMBA, hba} exploit the unimodal structure of the RSS in mmWave massive antenna systems using a MAB approach in cumulative regret minimization setting that balances exploration and exploitation. However, due to continuous exploration in regret settings, sub-optimal phase-shifts can be used for data transfer, resulting in outages. Therefore, regret is not the right performance measure. To overcome these issues, we address the beamforming problem of the HMT system in {\it pure-exploration setting}. 

The pure-exploration strategies proposed in \cite{ICC2021_HOSUB, TWC2022fastBAPureExploration, Arxiv2022mmwaveBAHierarchicalSubtreeElim} exploit the benefits of hierarchical codebooks and the unimodality structure of the beam signal strengths to achieve higher data rates in mmWave massive antenna systems.  
All these works are based on {\it fixed confidence setting} that uses the benefits of hierarchical codebooks. However, {\it pure-exploration fixed-budget setting} \cite{audibert2010best,JMLR2016_ComlexityofBestArmIdentification} are more suitable for the beamforming problem, as the exploration can be completed within the channel estimation phase. The {\it Unimodal Bandit for Best Beam (UB3)} algorithm developed in \cite{Arxiv2022ghoshub3} exploits the unimodal structure of the RSS of the mmWave massive antenna system and identifies the best beam with high probability within a fixed budget of channel estimation. Note that all the above-mentioned works are applied to mmWave massive antenna systems, and none of them are applied to the HMT systems. 

We develop a learning algorithm that identifies the best phase-shifts using {\it a fixed-budget pure-exploration setup} by exploiting the parametric and unimodal structure of the far-field channel gains of the HMT systems. 
To our knowledge, this has not been studied in HMT systems. 

The paper is organized as follows. The channel model for the HMT system is given in Sec. \ref{sec:setup}. The channel estimation strategy is formulated in Sec. \ref{sec:discretization}. The proposed algorithm for learning the optimal phase-shifts is given in Sec. \ref{sec:algo}. We provided the theoretical guarantee of the proposed algorithm in Sec. \ref{sec:analysis}. Numerical evaluation of the proposed algorithm is provided in Sec. \ref{sec:experiments}. Finally, Sec. \ref{sec:discussion} concludes the paper.

\section{Channel Model and Channel Estimation}
\label{sec:setup}
In this section, we consider the channel model  as given in \cite{Arxiv2022learningHMT} for which we propose our algorithm. We follow the setup and notation provided in \cite{ghermezcheshmeh2021channel,Arxiv2022learningHMT}.

\vspace{-1mm}
\subsection{Channel Model}
We consider an HMT of width $L_x$ and length $L_y$ comprised of $M$ number of sub-wavelength phase-shifting elements. We assume a LOS between an HMT and each user \cite{akdeniz2014millimeter}, and the non-line-of-sight (NLOS) components are incorporated into the noise model \cite{JSAC2014_MillimeterWaveChannelModel,TWC2021jointBT}. We consider that each user sends orthogonal pilots to the HMT, and thereby, we focus on the CSI estimation between the HMT and a typical user, see  Fig. \ref{fig:schematic_figure}.
\begin{figure}[ht]
\vspace{-3mm}
    \centering
    \includegraphics[height = 5cm, width = 7.5cm 
    ]{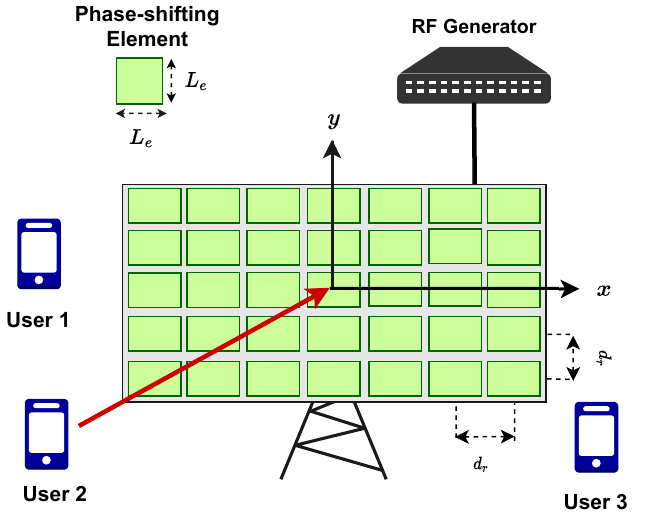} 
    \caption{\small{The HMT-assisted wireless communication system \cite{ghermezcheshmeh2021channel}.}}
    \label{fig:schematic_figure}
    \vspace{-6mm}
\end{figure}
Assuming that the HMT lies in the Cartesian coordinate system with the centre of the surface at the origin, let $\beta_{m_xm_y}$ be the phase-shift at the $(m_x, m_y)^{th}$ element, where $m_x \in \left \{ -\frac{(L_x/d_r)-1}{2},\dots,\frac{(L_x/d_r)-1}{2} \right \}$, $m_y \in \left \{ -\frac{(L_y/d_r)-1}{2},\dots,\frac{(L_y/d_r)-1}{2} \right \}$, and $d_r$ be the distance between two neighbouring phase-shifting elements. We denote $\lambda$ as the wavelength of the carrier frequency, $k_0=2\pi/\lambda$ be the wave number, $d_0$ be the distance between the user and the center of the HMT, and $F_{m_xm_y}$  be the effect of the size and power radiation pattern of the $(m_x,m_y)^{th}$ phase-shifting element on the channel coefficient \cite{ellingson2021path}. In the far-field, the radiation pattern of all the phase-shifting elements of the HMT is identical, i.e., $F_{m_x m_y} = F,\; \forall m_x, m_y$ holds \cite{ellingson2021path}. Finally, we denote $\theta$ and $\phi$ as the elevation and azimuth angles of the signals from the user to the centre of the HMT, refer to Fig. 2 in \cite{Arxiv2022learningHMT}.
We consider the phase-shift imposed by the $(m_x,m_y)^{th}$ element is set to
\begin{equation}
\label{eqn:PhaseShiftParamterizess}
\beta_{m_xm_y} = -\mod(k_0d_r(m_x\beta_1 + m_y\beta_2), 2\pi),
\end{equation}
where $\beta_1$ and $\beta_2$ are the phase-shifting parameters \cite{ghermezcheshmeh2021channel, selvan2017fraunhofer, najafi2020physics} and they are the only degrees of freedom in $\beta_{m_xm_y}$. With $\beta_{m_xm_y}$ as in \eqref{eqn:PhaseShiftParamterizess}, the channel between the HMT and the typical user \cite{Arxiv2022learningHMT, ghermezcheshmeh2021channel} for far-field regions is given by
\begin{align}
\label{eqn:contphaseHMMIMOuserchannel}
    H^{ff}(\beta_1, \beta_2) &= \left(\frac{\sqrt{F}\lambda e^{-jk_0d_0}}{4\pi d_0}\right)L_xL_y \mathrm{sinc}\big(K_x\pi(\alpha_1 - \beta_1)\big)\nonumber\\
&\qquad\times\mathrm{sinc}\big(K_y\pi(\alpha_2 - \beta_2)\big),
\end{align}
where $K_x = L_x/\lambda, K_y = L_y/\lambda, \alpha_1 = \sin(\theta)\cos(\phi), \alpha_2 = \sin(\theta)\sin(\phi),$ and  $\mathrm{sinc}(x) = \frac{\sin(x)}{x}$. Note that $\alpha_1 \in [-1,1]$ and $\alpha_2 \in [-1,1]$, and their values depend on the user's location, i.e., on $\theta$ and $\phi$. Therefore, $\alpha_1$ and $\alpha_2$ are the two unknown parameters that need to be learned by the HMT \cite{Arxiv2022learningHMT}. 
\begin{figure}[ht]
    \centering
    \vspace{-8mm}
\includegraphics[height = 4.6cm, width = 8cm, trim = {1 1 1 2}, clip]{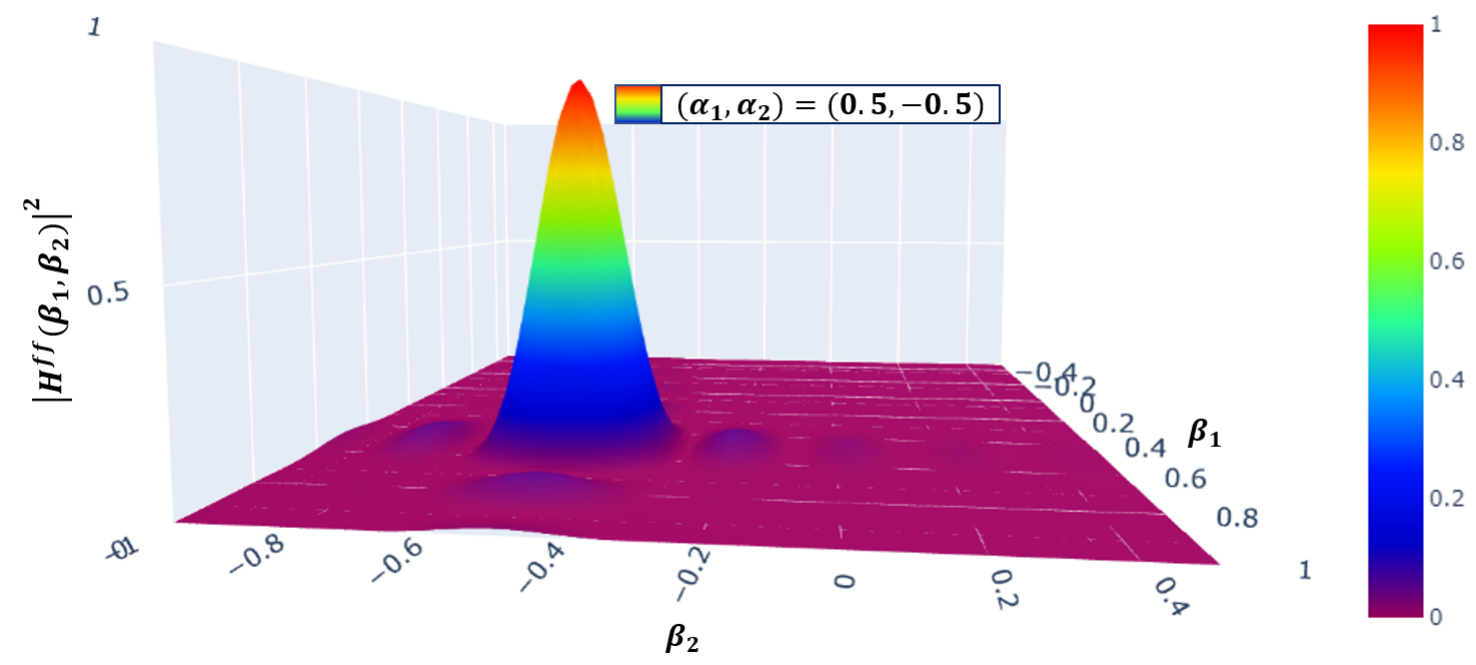}
    \caption{\small{The optimal value is attained at the highest peak of the central lobe of $|H^{ff}(\beta_1,\beta_2)|$, i.e. $(\beta^*_1,\beta^*_2) = (\alpha_1,\alpha_2) = (0.5, -0.5).$}} 
    \label{fig:abs_h}
    \vspace{-5mm}
\end{figure}

\subsection{Channel Estimation}
\label{sec:Objective}

The RSS at the HMT for fixed phase-shifting parameters $(\beta_1,\beta_2),$ denoted by $r(\beta_1,\beta_2)$,  is given by
 \begin{align}
    r(\beta_1,\beta_2) &= 
    \big| \sqrt{P}\times H^{ff}(\beta_1,\beta_2) + \zeta\big|^2, \label{eqn:absRSS}
\end{align}
which is induced when the user sends a pilot symbol $x_p = \sqrt{P}$ to the HMT, where $P$ is the pilot transmit power and $\zeta$  is the complex-valued additive white Gaussian noise (AWGN) with zero mean and variance $\sigma^2$ at the HMT. The mean RSS of $r(\beta_1,\beta_2)$, denoted by $\mu(\beta_1,\beta_2)$, is given by
\begin{align}
    \hspace{-1mm}\mu(\beta_1,\beta_2)= \EE{r(\beta_1,\beta_2)} = \abs{\sqrt{P}\times H^{ff}(\beta_1,\beta_2)}^2 + \sigma^2.\label{eqn:expectedRSS}
\end{align}
Our goal is to identify the optimal phase-shifting parameters that maximize the mean RSS $\mu(\beta_1,\beta_2)$, i.e., 
\begin{align}
\label{eqn:expobjfn}
    (\beta_1^*,\beta_2^*) &= \argmax_{\substack{\beta_1 \in [-1,1], \beta_2 \in [-1,1]}}\mu(\beta_1,\beta_2).
\end{align} 
Notice that $\mu(\beta_1,\beta_2)$ is maximized at the same points where the absolute value of the HMT-user channel given by \eqref{eqn:contphaseHMMIMOuserchannel} is maximized, i.e., $\beta_1^*=\alpha_1,\beta_2^*=\alpha_2$  (see Fig. \ref{fig:abs_h}). Note that $(\alpha_1, \alpha_2)$ are unknown and need to be learned.  Moreover, $\alpha_1$ and $\alpha_2$ take continuous values in the range $[-1,1]$ making it infeasible to learn their exact values within a finite number of pilot symbols. We thus quantize the range of the phase-shifting parameters and work with their discrete values. 

\section{Discretizating Phase-Shifting Parameters}
\label{sec:discretization}
In this section, we discretize the range of the phase-shifting parameters such that the optimal values of these parameters in the discrete set are not far from the optimal values in the continuous space, and the difference in their corresponding mean RSS value remains small.

\vspace{-1mm}
\subsection{Discretization of the Phase-Shifting Parameters}
We discretize the range of $\beta_1$ into $K_1$ equally spaced points. The set of discrete phase-shift parameters, denoted by $\mathcal{B}_1$, is
\begin{align}
    \label{eqn:Beta1_set}
    \mathcal{B}_1 = \{\beta^1_1,\beta^2_1,\dots,\beta^{K_1}_1\},
\end{align}
where $\beta^k_1 = \beta^1_1 + (k-1)d_1, k=1,2,\ldots,K_1$ and $d_1$ is a constant to be set later. Similarly, we discretize the range of $\beta_2$ into $K_2$ equally spaced points, denoted by $\mathcal{B}_2$, is 
\begin{align}
    \label{eqn:Beta2_set}
    \mathcal{B}_2 = \{\beta^1_2,\beta^2_2,\dots,\beta^{K_2}_2\},
\end{align}
where $\beta^k_2 = \beta^1_2+ (k-1)d_2, k=1,2,\ldots, K_2$ and $d_2$ is a constant to be set later. 

\textbf{Objective:} We revise objective in  Eq.~(\ref{eqn:expobjfn}) to the following by restricting to the discrete sets $\mathcal{B}_1$ and $\mathcal{B}_2$.
\begin{align}
\label{eqn:expobjfndiscete}
    (\beta_1^{k_1^*},\beta_2^{k_2^*}) &= \argmax_{\substack{\beta_1 \in \mathcal{B}_1,\beta_2 \in \mathcal{B}_2}} \mu(\beta_1,\beta_2).
\end{align} 
Here $k_1^*$ and $k_2^*$ denote the index of the optimal values in the set $\mathcal{B}_1$ and $\mathcal{B}_2$, respectively. We construct $\mathcal{B}_1$ and $\mathcal{B}_2$ such that the mean RSS achieved from \eqref{eqn:expobjfn} remains close to that obtained from \eqref{eqn:expobjfndiscete}. To demonstrate this, we use the following notations and unimodality property of the mean RSS in each argument over the discrete spaces. 

For a fixed $\beta^0_2 \in [-1,1]$ we write $r_1(\beta_1):=r(\beta_1,\beta^0_2)$, $\mu_1(\beta_1):=\mu(\beta_1,\beta^0_2)$ for all $\beta_1 \in \mathcal{B}_1$. Similarly, for a fixed $\beta^0_1 \in [-1,1]$ we write $r_2(\beta_2):=r(\beta^0_1,\beta_2)$, $\mu_2(\beta_2)=\mu(\beta^0_1,\beta_2)$ for all $\beta_2 \in \mathcal{B}_2$, . We set $\beta_1^1=\beta_2^1=-1$ and $\beta_1^{K_1}=\beta_2^{K_2}=1$ in $\mathcal{B}_1$ and $\mathcal{B}_2$.


\begin{defi}[Unimodality]
For any discrete set $\mathcal{B}=\{b_1, b_2, \ldots, b_K\}$,
we say that a real-valued function $f:\mathcal{B}\rightarrow \mathcal{R}$, is unimodal if there exits an $1<i^* < K$ such that 
\begin{align}
\label{eqn:unimodaldefn}
    f(b_1)<f(b_2)<\dots <f(b_{i^*})>..>f(b_{K-1})>f(b_K)
\end{align}
\end{defi}
\noindent
The following lemma established the unimodal property of the function $\mu_1: \mathcal{B}_1\rightarrow \mathcal{R}$ for any $\beta_2^0 \in [-1,1]$. For notational convenience we write  $G:= \big|\sqrt{P}\left(\frac{\sqrt{F}\lambda e^{-jk_0d_0}}{4\pi d_0}\right)L_xL_y\big|^2$. 
\begin{lem}
\label{lem:discreteBeta1}
Let $\mathcal{B}_1$ be such that $K_1 = \ceil{2/d_1}+1$ and $d_1 = \frac{1}{K_x}$, then $\mu_1(\beta_1)$ is unimodal on $\mathcal{B}_1$. Similarly, let $K_2 = \ceil{2/d_2}+1$ and $d_2 = \frac{1}{K_y}$ in $\mathcal{B}_2$, then $\mu_2(\beta_2)$ is unimodal on $\mathcal{B}_2$. 
\end{lem}
\begin{proof}
The proof is given in Sec. \ref{subsec:ProofLemma1}.
\end{proof}
For $\beta_2$ fixed at $\beta^0_2$, every side lobe of $|H^{ff}(\beta_1, \beta^0_2)|$ hits zero at the periodicity of $\pi$, and the central lobe hits zero at the periodicity of $2\pi$. When we discretize the phase-shifting parameter $\beta_1$ as in Lemma \ref{lem:discreteBeta1}, we get one discrete value in each side lobe and two values in the central lobe of $|H^{ff}(\beta_1, \beta^0_2)|$, (see Fig. \ref{fig:disccontmeanRSS}). One of these values in the central lobe is closest to $\alpha_1$, given by $\beta^{k_1^*}_1$. Moreover, if one of the values in the central lobe hits zero, the other will be the optimal value, i.e., $\alpha_1$. The following proposition bounds the difference in the optimal mean RSS obtained from (\ref{eqn:expobjfn}) and (\ref{eqn:expobjfndiscete}). 
\begin{prop}
\label{cor:reducedmeanRSS}
Consider $\mathcal{B}_1$ and $\mathcal{B}_2$ be as specified in Lemma \ref{lem:discreteBeta1}. 
Let $|\beta^{k_1^*}_1 - \alpha_1| =\epsilon_1$ and $|\beta^{k_2^*}_2 - \alpha_2|=\epsilon_2$. Then,
\begin{align}
&\abs{\mu(\alpha_1,\alpha_2) -  \mu(\beta^{k_1^*}_1, \beta^{k_2^*}_2)} \nonumber\\
&= G
\big(1 - \big|\mathrm{sinc}\big( K_x\pi\epsilon_1\big)\mathrm{sinc}\big(K_y\pi \epsilon_2\big)\big|^2\big).\label{eqn:epsilondiff}\\
& \leq G\left[1 - \left(2/\pi\right)^4\right].\label{eqn:epsilondiffbound}
\end{align}
\end{prop}
\begin{proof}
The proof is given in Sec. \ref{subsec:ProofCorollary1}.
\end{proof}
Note that as $\epsilon_1,\epsilon_2 \to 0$, $\mu(\beta^{k_1^*}_1, \beta^{k_2^*}_2) \to \mu(\alpha_1, \alpha_2)$ in \eqref{eqn:epsilondiff}. Further, the proof establishes that $\epsilon_1 \leq 1/(2K_x)$ and $\epsilon_2 \leq 1/(2K_y)$, to obtain the bound  (\ref{eqn:epsilondiffbound}). The worst-case scenario occurs when the mean RSS values of both discrete parameters in the central lobe are the same and result in the maximum error, i.e., $\epsilon_1 = 1/(2K_x)$ and $\epsilon_2 = 1/(2K_y)$.
\begin{figure}[ht]
    \centering
    \includegraphics[width = 8cm, height = 4.8cm, trim = {0.1cm 0cm 0cm 0.6cm}, clip]{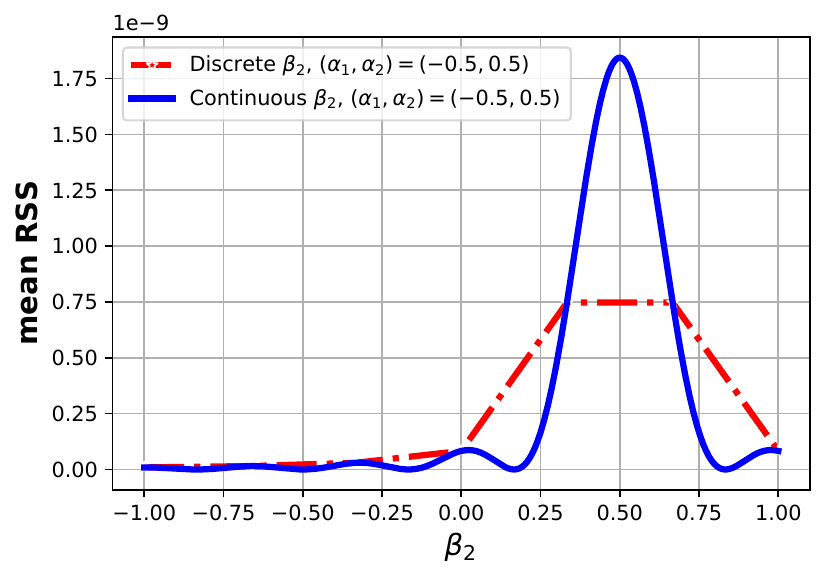}
    \caption{\small{Mean RSS vs Discrete and Continuous $\beta_2$}, where $\beta^0_1= -0.5$.}
    \label{fig:disccontmeanRSS}
    \vspace{-3mm}
\end{figure}

\subsection{Problem Formulation}
\label{subsec:probform}
We assume a block fading channel model where the duration of each block is denoted as $T=T_E+T_D$. We use duration $T_E$ to estimate the phase-shifting parameters and the remaining $T_D$ for data transmission. We consider a slotted system where $n$ pilot symbols are used over $T_E$. Each pilot symbol is assumed to be of one slot duration. We consider an additive white Gaussian noise channel, where HMT observes noisy RSS values in each slot. For each phase-shifting parameter, the observations are independently and identically distributed.

 We address the channel estimation problem as a fixed-budget pure-exploration MAB approach \cite{audibert2010best, ICML2012_PACSubsetSelection}. In a MAB framework, a policy is any strategy that selects the best phase-shifting parameter value based on noisy RSS observed in the previous slots.
 Let $\mathbf{\Pi}$ denote the set of pure exploration strategies that output the best parameter values using $n$ pilots symbols. For any policy $\pi \in \Pi$, let $(\beta^{\hat{k}_1^\pi}_1, \beta^{\hat{k}_2^\pi}_2) \in \mathcal{B}_1\times \mathcal{B}_2$ denote the parameters output by the policy. Our objective is to learn a policy that minimizes the probability of misidentifying the optimal phase-shifting parameters as given by
 \begin{align}
\label{eqn:errorprob}
\min_{\pi \in \mathbf{\Pi}}\Pr \bigg (\bigcup\limits_{i=1,2} \Big\{\beta^{\hat{k}_i^\pi}_i \neq \beta^*_i\Big\}\bigg),
\end{align}
where $\Pr(\cdot)$ is calculated with respect to the samples induced by the policy.  
\vspace{-3mm}
\section{Proposed Algorithm} 
\label{sec:algo}
In this section, we develop an algorithm that minimizes the probability of misidentifying the optimal phase-shifting parameters with a budget of $n$. Note that the channel gain function  has a variable separable in the phase-shifting parameters, i.e., each $\mathrm{sinc}$ function in \eqref{eqn:contphaseHMMIMOuserchannel} depends only on one of the parameters. Hence, we can learn one of the parameters while fixing the other. 



We propose the \Algo algorithm that finds the best phase-shift parameters by exploiting the unimodal structure of mean RSS using $n$ pilot symbols. The {\it \HB } works in two phases, each consisting of $n/2$ pilot symbols. In first phase, {\it \HB} fixes $\beta_2$ to $\beta^0_2 \in \mathcal{B}_2$ and runs a baseline algorithm {\it \AlgoI } that outputs $\beta^{\hat{k}_1}_1 \in \mathcal{B}_1$.  In second phase, {\it \HB} fixes $\beta_1$ to $\beta^{\hat{k}_1}_1 \in \mathcal{B}_1$ and runs {\it \AlgoII } algorithm that output $\beta^{\hat{k}_2}_2 \in \mathcal{B}_2$. The pseudo-code of {\it \HB} is given in {\it ALGO 1}. 
\begin{algorithm}[h]
	\renewcommand{\thealgorithm}{ALGO 1: \Algo}
	\floatname{algorithm}{}
	\caption{\bf }
	\label{algo:PSHMT}
    \begin{algorithmic}[1]
        \STATE \textbf{Input:} $n$, $\beta^0_2$, $\mathcal{B}_1$ and $\mathcal{B}_2$. 
        
        \textbf{\ding{72}\ding{72}  Phase 1: Best Phase-Shift Estimation for $\beta_1$  \ding{72}\ding{72}}
        \STATE Fix $\beta_2$ at $\beta^0_2$. Set $\mathcal{D}_1 = \mathcal{B}_1$.
        \STATE Run \AlgoI$\big(n/2, \mathcal{D}_1\big)$, and $\beta^{\hat{k}_1}_1$ is the output of Phase 1 after $L_1+1$ batches.

        \textbf{\ding{72}\ding{72}  Phase 2: Best Phase-Shift Estimation for $\beta_2$  \ding{72}\ding{72}}
        \STATE Fix $\beta_1$ at $\beta^{\hat{k}_1}_1$. Set  $\mathcal{D}_2 = \mathcal{B}_2$.
        \STATE Run \AlgoII$\big(n/2, \mathcal{D}_2\big)$, and $\beta^{\hat{k}_2}_2$ is the output of Phase 2 after $L_2+1$ batches.
       \STATE \textbf{Output:} $\beta^{\hat{k}_1}_1$ and $\beta^{\hat{k}_2}_2$.
        \STATE  \textbf{Phase-shifts at HMT: } Set the phase-shift of the $(m_x,m_y)^{th}$ element at the HMT to
        \vspace{-2mm}
      $$\beta_{m_xm_y} = -\mod\big(k_0d_r(m_x\beta^{\hat{k}_1}_1 + m_y\beta^{\hat{k}_2}_2),2\pi\big).$$
    \end{algorithmic}
\end{algorithm}
\vspace{-1mm}
\subsection{The Baseline Algorithm: \Algoi}
The {\it \HB} runs {\it \Algoi} algorithm in phase $i$ for $i=1,2$. In phase $i$, {\it \Algoi} exploits the unimodal structure of the mean RSS. 
It divides the $n/2$ pilot symbols into $L_i+1$ batches (see \ref{eqn:L}). The $l$-th batch contains $N^l_i$ pilots, where $N^l_i$ is set as
\begin{equation}
\label{eqn:Nl}
    N^l_i=\begin{cases}
        \frac{2^{L_i-2}}{3^{L_i-1}}\frac{n}{2} & \mbox{ for  } l=1,2\\
        \frac{2^{L_i-(l-1)}}{3^{L_i-(l-2)}}\frac{n}{2} & \mbox{ for  }  l=3,4,\ldots, L_i+1
    \end{cases}
\end{equation}
We set $N^l_i$ as in \eqref{eqn:Nl} so that $N^l_i$ satisfy the following
\begin{align}
\label{eqn:L_T}
    \sum\limits_{l=1}^{L_i+1}N^{l}_i = 2\times \frac{2^{L_i-2}\frac{n}{2}}{3^{L_i-1}} + \sum_{l=3}^{L_i+1}\frac{2^{L_i-(l-1)}\frac{n}{2}}{3^{L_i-(l-2)}} = \frac{n}{2}.
\end{align} 
After $L_i+1$ batches, {\it \Algoi} outputs $\beta^{\hat{k}_{i}}_i \in \mathcal{B}_i$, which we declare as the best phase-shift parameter value for $\beta_i$.

The pseudo-code of {\it \Algoi} is given in ALGO 2. It works as follows. The set $\mathcal{A}_l$ denotes the set of discrete parameter values for $\beta_i$ in batch $l$, and $j_l:= |\mathcal{A}_l|$ denotes the number of values in the set $\mathcal{A}_l$. In each batch $l=1,2,\ldots L_i$, HMT uses $N^l_i/4$ number of pilots for each value in $\mathcal{S}_l = \{\beta^{k^A}_i, \beta^{k^B}_i, \beta^{k^C}_i, \beta^{k^D}_i\} \subset \mathcal{A}_l$, where $\mathcal{S}_l$ is selected such that it include the two extremes and two middle parameter values uniformly spaced from them (lines $4$-$7$). At the end of batch $l$, HMT obtains the empirical means $\hat{\mu}^l_{\beta^{k}_i}$ for each $\beta^{k}_i \in \mathcal{S}_l$ (line $8$). Based on these empirical means, the algorithm eliminates at most $1/3^{rd}$ of the number of values from the set $\mathcal{A}_l$\footnote{If the number of parameters in a batch is not a multiple of $4$, then less than $1/3^{rd}$ will be eliminated in that batch.}. Specifically, if $\beta^{k^A}_i$ or $\beta^{k^B}_i$ has the highest empirical means, then all the values succeeding $\beta^{k^C}_i$ in $\mathcal{A}_l$ are eliminated (line $11$).  Similarly, if $\beta^{k^C}_i$ or $\beta^{k^D}_i$ have the highest empirical means, then all the values preceding $\beta^{k^B}_i$ in $\mathcal{A}_l$ are eliminated (line $13$). Fig.~\ref{fig:AlgorithmPic} gives a pictorial representation of eliminating these values in the two cases. The remaining set of values is then transferred to the next batch. In batch $L_i+1$, we are left with three parameter values. Each is sampled $N^{L_i+1}_i/3$ times, and the one with the highest empirical mean is the output as the best parameter value of $\beta_i$ (lines $17$-$22$).
\begin{algorithm}[h]
\renewcommand{\thealgorithm}{ALGO 2: \Algoi} 
	\floatname{algorithm}{}
	\caption{\bf }
	\label{algo:UB3}
    \begin{algorithmic}[1]
        \STATE \textbf{Input:} $n/2$ and $\mathcal{D}_i$.
        \STATE \textbf{Initialise:} $\mathcal{A}_1 = \mathcal{D}_i$,  $j_1\leftarrow |\mathcal{A}_1|$. Calculate $L_i$ from \eqref{eqn:L}.
        \FOR{$l= 1$ to $L_i$}
            \STATE $\beta^{k^A}_i\leftarrow$ First phase-shift parameter of $\mathcal{A}_l$.
              \STATE $\beta^{k^B}_i \leftarrow \lceil j_l/3\rceil^{th}$ phase-shift parameter of $\mathcal{A}_l$.
              \STATE $\beta^{k^C}_i \leftarrow \lfloor 2j_l/3\rfloor^{th}$ phase-shift parameter of $\mathcal{A}_l$.
               \STATE      $\beta^{k^D}_i\leftarrow$ Last phase-shift parameter of $\mathcal{A}_l$.
                \STATE HMT collects $\frac{N^l_i}{4}$ pilots for each value in $\mathcal{S}_l = \{\beta^{k^A}_i, \beta^{k^B}_i, \beta^{k^C}_i, \beta^{k^D}_i\}$ and 
                 compute empirical means denoted as $\hat{\mu}^l_{\beta^{k^A}_i}$, $\hat{\mu}^l_{\beta^{k^B}_i}$, $\hat{\mu}^l_{\beta^{k^C}_i}$, $\hat{\mu}^l_{\beta^{k^D}_i}$. 
                \STATE $x^*_l = \argmax_{\beta^{k}_i \in \mathcal{S}_l} \hat{\mu}^l_{\beta^{k}_i}.$
                \IF{$\beta^{x^*_l}_i == \{\beta^{k^A}_i, \beta^{k^B}_i\}$}
                    \STATE $\mathcal{A}_{l+1} \leftarrow\{\beta^{k}_i \in \mathcal{A}_l: \beta^{k^A}_i \leq \beta^k_i \leq \beta^{k^C}_i\}$ 
                \ELSIF{$\beta^{x^*_l}_i == \{\beta^{k^C}_i, \beta^{k^D}_i\}$}
                     \STATE $\mathcal{A}_{l+1}$ $\leftarrow\{\beta^{k}_i \in \mathcal{A}_l: \beta^{k^B}_i \leq \beta^{k}_i \leq \beta^{k^D}_i\}$ 
                \ENDIF
                \STATE $j_{l+1} \leftarrow |\mathcal{A}_{l+1}|;$
               \ENDFOR
        \FOR{$l=L_i+1$}
            \STATE $\mathcal{A}_{L_i+1} = \{\beta^{k^A}_i, \beta^{k^B}_i, \beta^{k^C}_i\}$;
            \STATE HMT collects $\frac{N^{L_i+1}_i}{3}$ pilots for each value in $\mathcal{A}_{L_i+1}$, and obtain $\hat{\mu}^l_{\beta^{k^A}_i}$, $\hat{\mu}^l_{\beta^{k^B}_i}$, $\hat{\mu}^l_{\beta^{k^C}_i}$.\\ 
            \STATE Obtain $\beta^{\hat{k}_i}_i=\argmax_{k \in \mathcal{A}_{L_i+1}}\hat{\mu}_k.$
        \ENDFOR
       \STATE \textbf{Output:} $\beta^{\hat{k}_i}_i$
    \end{algorithmic}
\end{algorithm}
\begin{figure}
\vspace{-5mm}
    \centering
    \includegraphics[width = 6cm, height = 6cm, trim = 1cm 0.4cm 1.2cm 0.3cm, clip]{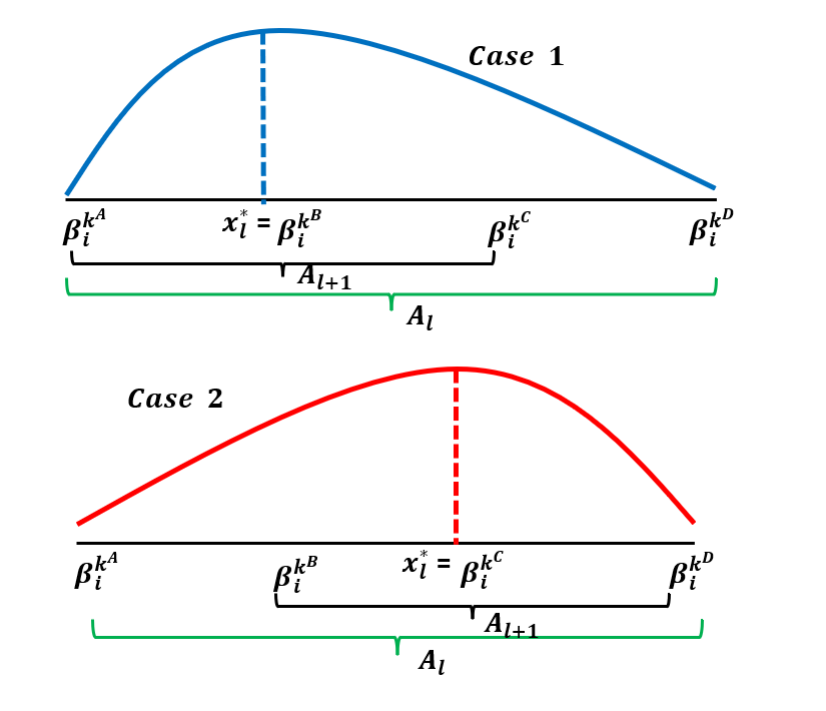}
    \caption{Different cases of elimination in batch $l$.}
    \label{fig:AlgorithmPic}
    \vspace{-5mm}
\end{figure}

Note that the values between $\beta^{k^A}_i$ and $\beta^{k^B}_i$ or the values between $\beta^{k^C}_i$ and $\beta^{k^D}_i$ are eliminated in each batch, and the values between $\beta^{k^B}_i$ and $\beta^{k^C}_i$ always survive for $i=1,2$. After batch $l=1,2,\ldots, L_i$, only $\lfloor\frac{2}{3}j_l\rfloor$ of the values survive. For ease of exposition, we will drop the $\lfloor\rfloor$ function since this drop will influence only a few constants in the analysis. Hence, based on the elimination strategy and $N^l_i$ as given in \eqref{eqn:Nl}, after the end of $L_i$ batches, there will be three parameter values, i.e. $\left(2/3\right)^{L_i}K_i =3$. Therefore, we get 
\begin{align}
\label{eqn:L}
L_i =\frac{\log_2 K_i/3}{\log_2 3/2}, \text{ for } i=1,2.
\end{align}
\section{Analysis}
\label{sec:analysis}
In this section, we find an upper bound on the error probability of {\it \HB} 
for fixed-budget pure exploration bandits with unimodal structure. 
\begin{thm}
\label{thm:PSHMTBound}
Consider $\mathcal{B}_1$ and $\mathcal{B}_2$ as given in Lemma \ref{lem:discreteBeta1}. In phase $i$, {\it \HB } runs {\it \Algoi } in $L_i+1$ batches, where  $L_i=\frac{\log_2 K_i/3}{\log_2 3/2}$, and $K_i = \abs{\mathcal{B}_i}$, $i= 1,2$. Let $\beta^{\hat{k}_i}_i$ be the output after phase $i$, and $\Delta_i = \min\limits_{2\leq k \leq K_i-1}\big|\mu_i(\beta^k_i) - \mu_i(\beta^{k-1}_i)\big|>0$ denotes the minimum gap between the mean RSS of any two neighbouring parameter values of $\beta_i$.  Then, for any $n > (K_1 + K_2)$, the error probability of {\it \HB} 
is upper bounded as 
\begin{align}
\label{eqn:PSHMTBound}
 &\Pr \bigg (\bigcup\limits_{i=1,2} \Big\{\beta^{\hat{k}_i^\pi}_i \neq \beta^*_i\Big\}\bigg)\nonumber\\
   &\le 4\sum\limits_{i=1}^2\bigg[\left(\frac{\log_2 K_i/3}{\log_2 3/2}-1\right)\exp\left\{-\frac{n\Delta_i^2}{1296\sigma^4(1+3G/\sigma^2)^4}\right\}\nonumber\\
   &\qquad \qquad +2\exp\left\{-\frac{nK_i\Delta_i^2}{1296\sigma^4(1+3G/\sigma^2)^4}\right\}\bigg].
\end{align}
\end{thm}
\begin{proof}
The proof is in Sec. \ref{subsec:ProofThm1}.
\end{proof}
\vspace{-2mm}
As the first term is dominant in \eqref{eqn:PSHMTBound}, the error probability of {\it \HB} in \eqref{eqn:PSHMTBound} is thus of order $\mathcal{O}\left(\sum\limits_{i=1}^2\log K_i\exp\left\{-\frac{n\Delta_i^2}{\sigma^4(1+3G/\sigma^2)^4}\right\}\right)$. For  unstructured bandits, the error probability of {\it Seq. Halv.} algorithm is of order  $O\left(\sum\limits_{i=1}^2\log K_i\exp\left\{-\frac{n/H_i}{\log K_i\sigma^4(1+3G/\sigma^2)^4}\right\}\right)$ which matches with the lower bound up to a multiplicative factor of $\log_2(K_i)$ \cite{carpentier2016tight}, where $H_i$ is the complexity parameter that depends on the sub-optimality gap of $\beta_i$. As expected, the error probability of {\it \HB} is smaller than that of {\it Seq. Halv.}, as the $\log K_i$ factor present in the denominator of exponent of the error bound of {\it Seq. Halv.} does not appear in the exponent of the error bound of {\it \HB}. 

\section{Numerical Simulations}
\label{sec:experiments}
For the numerical evaluation, we set $L_x=L_y=1$m, carrier frequency is 30 GHz, $\lambda = 1$cm, $d_r=\lambda/4$, $F = 1.6d_r/\lambda$, and the noise power for $200$ KHz ($\sigma^2$) is $-115$ dBm as in \cite{zhang2022beam}. 
We compare {\it \HB} with the following algorithms:

{\bf Iterative Algorithm (Iterative Algo.) \cite{ghermezcheshmeh2021channel}} estimates the best phase-shifts based on stopping criteria. The algorithm is based on the initial value of $(\beta^0_1, \beta^0_2)$ \cite[Sec. V.C]{ghermezcheshmeh2021channel}. {\bf Two-Stage Phase-Shifts Estimation Algorithm (Two-Stage Algo.) \cite{Arxiv2022learningHMT}} is based on pure-exploration which performs better than {\it Iterative Algo.} \cite{ghermezcheshmeh2021channel}. {\bf Sequential Halving (Seq. Halv.) \cite{seqhalving} } is a fixed budget pure-exploration algorithm for unstructured bandits, which is optimal by \cite{carpentier2016tight}. {\bf Linear Search Elimination (LSE) \cite{ICML2011_UnimodalBandits}} is a well-known algorithm for unimodal bandits proposed for continuous arms. We have considered it for a fixed $T_E$ budget and discrete points. {\bf HBA \cite{hba} } algorithm performs well for regret minimization that exploits the unimodal structure of mean RSS. The algorithm parameters are kept at $\rho_1=3$, $\gamma=0.5$. {\bf HOSUB \cite{ICC2021_HOSUB} } is a fixed-confidence pure-exploration based algorithm that exploits the unimodality of mean RSS. 

Note that we have adopted all these algorithms, except {\it Iterative Algo.} and {\it Two-Stage Algo.}, to fixed-budget setting, which works in two phases, where phase $i$ identifies the best parameter value of $\beta_i$. We ran the simulation $1000$ times and average values are plotted with $95\%$ confidence intervals. 
\begin{figure*}[ht]
	\centering
	\begin{subfigure}{0.45\textwidth}
		\includegraphics[width = 8.3cm, height = 5cm]{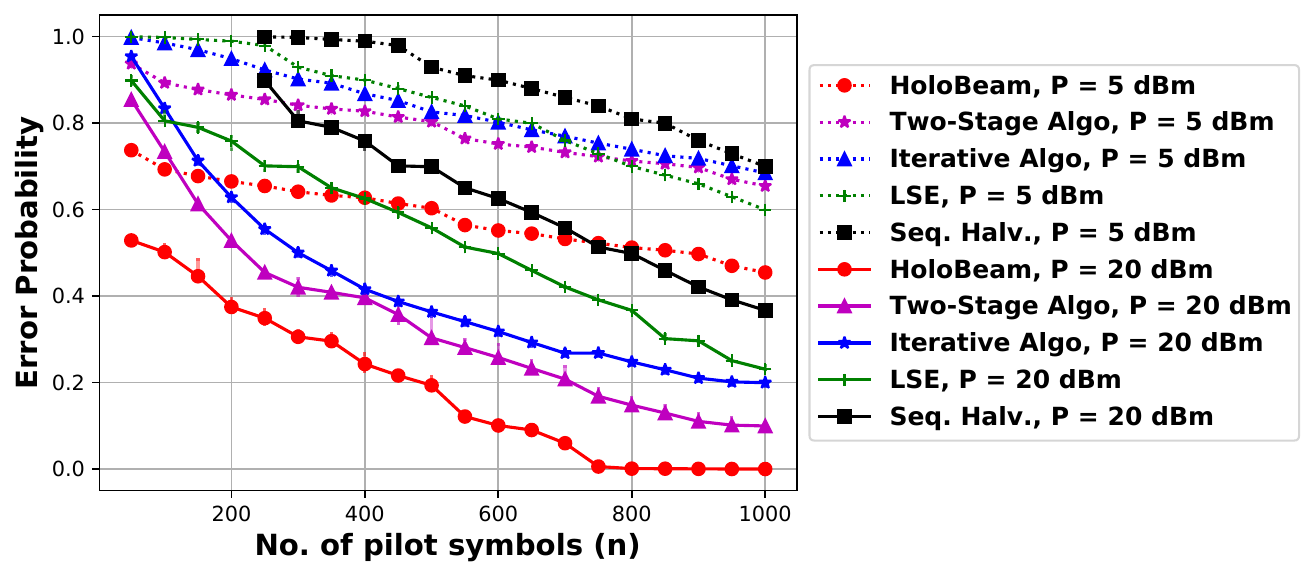}
  	\caption{Error Prob. vs No. of pilot symbols $(n)$ for \\ $d_0 = 800$m and $P = \{5, 20\}$dBm.}
    \label{fig:ErrorvsTimePower}
     \end{subfigure}
 	\begin{subfigure}{0.45\textwidth}
		\includegraphics[width = 8.3cm, height = 5cm]{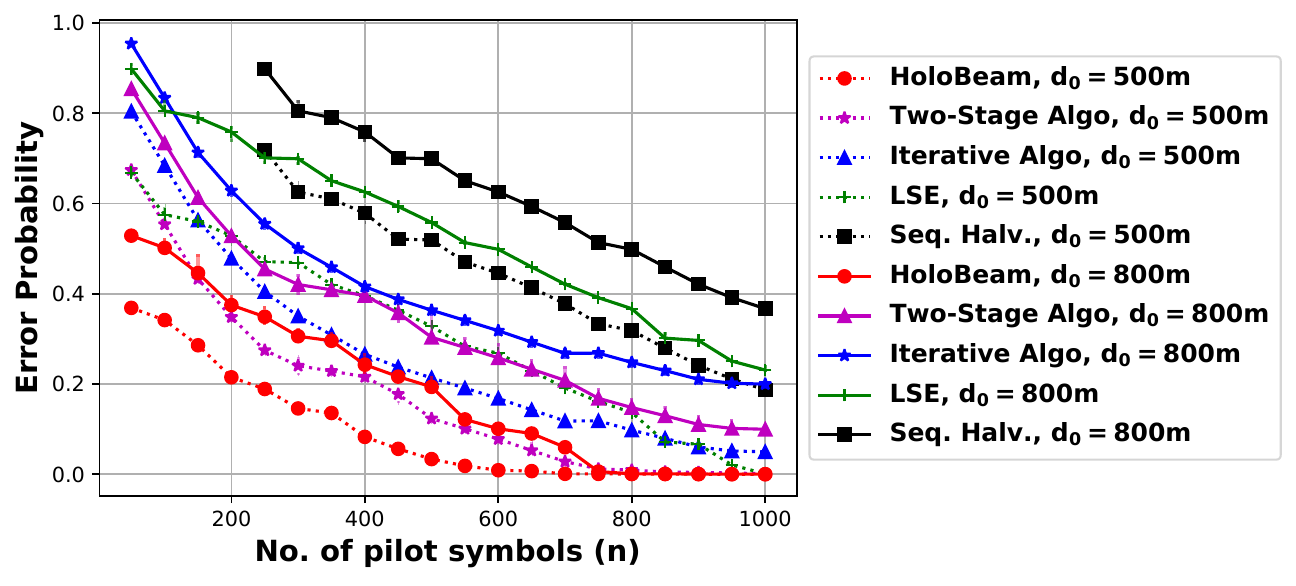}
    	\caption{Error Prob. vs No. of pilot symbols $(n)$ for \\$P = 20$ dBm and $d_0 = \{800, 500\}$ m.}
    	\label{fig:ErrorvsTimed20010}
	\end{subfigure}
    \caption{\small{Error probability performance of {\it \HB } vs State-of-the-art Algorithms}} 
    \label{fig:ErrorvsTime}
    \vspace{-4mm}
\end{figure*}
\begin{figure*}[ht]
	\centering
	\begin{subfigure}{0.45\textwidth}
		\includegraphics[width = 8.3cm, height = 5cm]{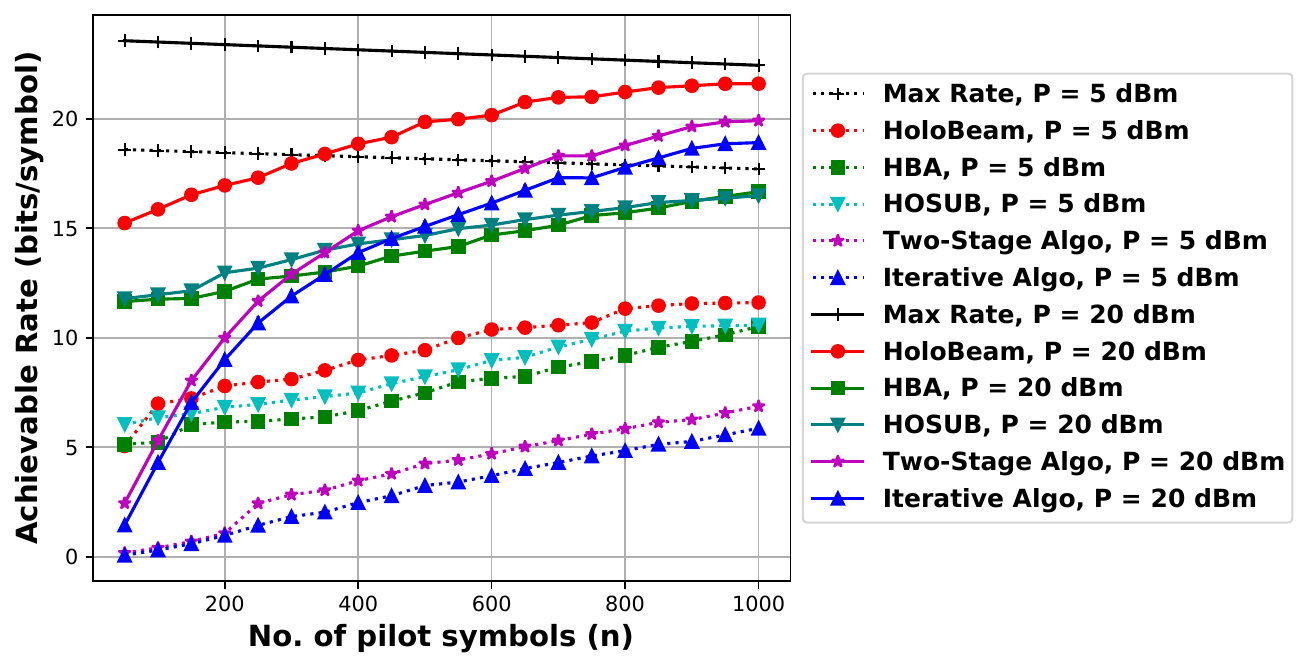}
    	\caption{Achievable rate vs No. of pilot symbols $(n)$ \\for transmit power $P =\{5, 20\}$ dBm.}
    	\label{fig:Throughput_vs_t}
	\end{subfigure}
		\begin{subfigure}{0.45\textwidth}
		\includegraphics[width = 8.3cm, height = 5cm]{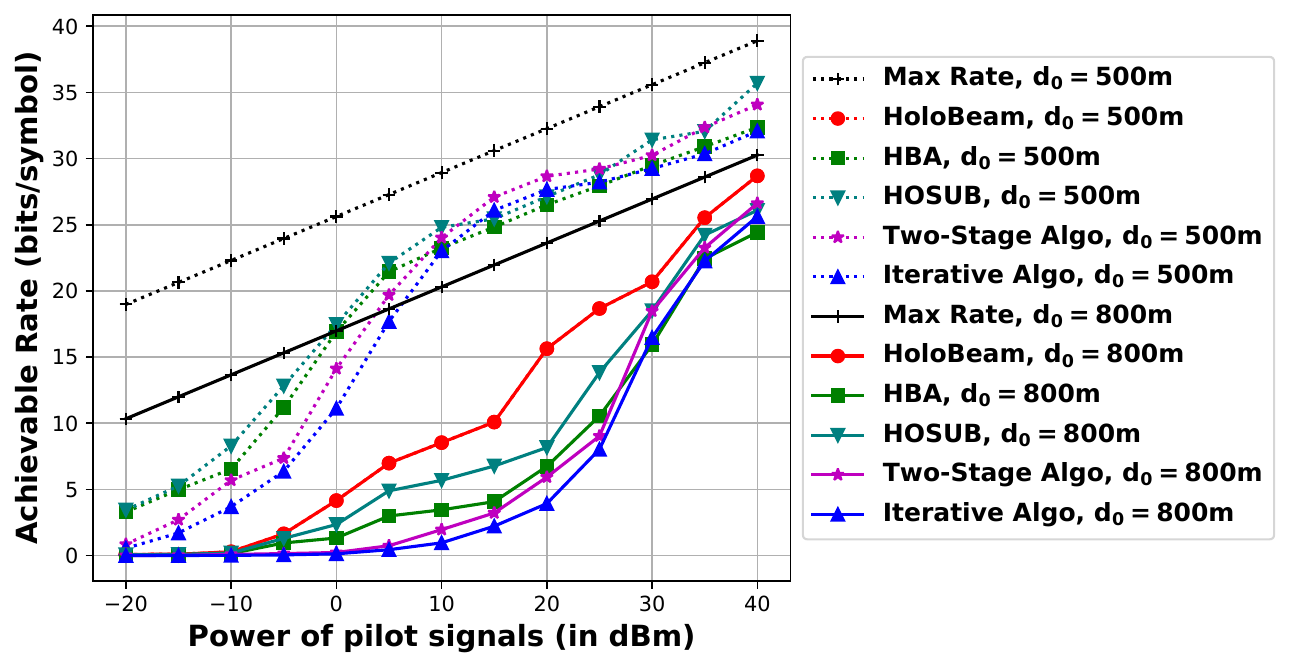}
    	\caption{Achievable rate vs transmit power of the pilot \\signals (in dBm) for $d_0 = \{500,800\}$ m and $n = 100$.}
    \label{fig:Throughput_vs_power}
	\end{subfigure}
 	\caption{Throughput performance of {\it {\it \HB}} vs State-of-the-art Algorithms}
	\label{fig:throughput}
\vspace{-4mm}
\end{figure*}
\subsection{Comparison with other pure exploration algorithms}
We compare the error probability performance for  {\it \HB} with the state-of-the-art algorithms used for pure-exploration, refer Fig. \ref{fig:ErrorvsTime}. 
When we increase the power of the pilot signals, the accuracy of the estimation of $(\alpha_1, \alpha_2)$ increases as received signals become less noisy. Thereby, the error probability decreases, refer to Fig. \ref{fig:ErrorvsTimePower}. Moreover, as distance increases, the error probability of the algorithms increases, refer to Fig. \ref{fig:ErrorvsTimed20010}. 

For transmit power $P = 20$ dBm, {\it \HB} can identify the best phase-shifts with a probability of more than 80\% within $200$ pilot symbols. In contrast, the other state-of-the-art algorithms identify the best phase-shifts with a probability of at least $45\%$. Furthermore, {\it Seq. Halv.} needs at least $K_i\log_2(K_i)$ batches (i.e. at least $250$ batches) to complete phase $i$, $i=1,2$. Hence, fewer pilots will remain for the algorithm to execute in the neighbourhood of the optimal phase-shifts compared to {\it \HB}. This demonstrates the advantage of exploiting the unimodality of the mean RSS.

\subsection{Comparison of throughput performance}
Note that {\it Seq. Halv.} and {\it LSE} are not included for throughput comparison due to their poor performance, as seen in Fig. \ref{fig:ErrorvsTime}. We have also considered the oracle scheme where $(\alpha_1,\alpha_2)$ are estimated perfectly and achieve the maximum rate. We define the throughput as the achievable data rates at the user obtained by the best estimated phase-shifts at the HMT after using $n$ number of pilots, as given by
\begin{align}
   \label{eqn:achievablerates}
&R\big(\beta^{\hat{k}_1}_1, \beta^{\hat{k}_2}_2\big) = \frac{T_D}{T}\log_2\bigg(1 +
\frac{P\big|H^{ff}(\beta^{\hat{k}_1}_1, \beta^{\hat{k}_2}_2)\big|^2}{\sigma^2}\bigg),
\end{align}
where $P$ is the transmission power at the HMT and $T_D$ is the duration left for data transmission. 


As the number of pilots increases, the channel estimation time increases, decreasing the data transmission time. At the same time, as we increase the number of pilots, the estimation accuracy of $(\alpha_1,\alpha_2)$ increases and the average throughput increases.  This trade-off in Eq. \ref{eqn:achievablerates} is captured in Fig. \ref{fig:Throughput_vs_t}. In Fig. \ref{fig:Throughput_vs_power}, with an increase in the power of the pilot signals, 
the average throughput increases, but decreases significantly as the user moves far from the HMT.  Moreover, {\it \HB} can improve the average throughput by more than 30\% compared to the state-of-the-art algorithms within $500$ pilot symbols. 

\section{Conclusion and future work}\label{sec:discussion}
We investigated learning optimal beamforming in an HMT-assisted wireless system using the fixed-budget pure-exploration framework. We exploited the parametric nature and unimodality of the channel gains to develop an algorithm named {\it \HB}. {\it \HB} identified the best phase-shifting parameters using a given number of pilots, and the error probability decays exponentially in the number of pilots. Simulations validated the efficiency of {\it \HB} as compared to the state-of-the-art algorithms. However, {\it \HB} works well when only the LOS path is present satisfying unimodal property. However, when NLOS paths are present, we face multimodal functions. It is interesting to extend {\it \HB} or develop new algorithms that adapt to channel gains with multiple modes.  

\section{Appendix}
\label{sec:appendix}
In this section, we will provide proof of the main results.
\vspace{-1mm}
\subsection{Proof of Lemma \ref{lem:discreteBeta1}}
\label{subsec:ProofLemma1}
\begin{proof}
\vspace{-1mm}
   By using \eqref{eqn:expectedRSS} and \eqref{eqn:contphaseHMMIMOuserchannel} for any $\beta^i_1 \in \mathcal{B}_1$, we obtain
\begin{align*}
    &\mu_1(\beta^i_1) - \sigma^2 = G \big|\mathrm{sinc}\big(K_x\pi(\alpha_1 - \beta^i_1)\big)\\
    &\qquad \times\mathrm{sinc}\big(K_y\pi(\alpha_2 - \beta^0_2)\big)\big|^2 
\end{align*}
For $\beta^{i+1}_1 = \beta^i_1 + d_1 \in \mathcal{B}_1$, we have the following
\begin{align*}
    &\mu_1(\beta^{i+1}_1) - \sigma^2 \nonumber\\
    &= G\big|\mathrm{sinc}\big(K_x\pi(\alpha_1 -\beta^{i+1}_1)\big)\mathrm{sinc}\big(K_y\pi(\alpha_2 - \beta^0_2)\big)\big|^2.
\end{align*}
For $d_1 = 1/K_x$,  $|\sin\big(K_x\pi(\alpha_1 - \beta^i_1 - d_1)\big)|$ $= |\sin\big(K_x\pi(\alpha_1 - \beta^i_1)\big)|$. Applying this fact, we get
\vspace{-1mm}
\begin{align}
    \frac{\mu_1(\beta^i_1) - \sigma^2}{ \mu_1(\beta^{i+1}_1) - \sigma^2} = \abs{\frac{\alpha_1 - \beta^i_1 - d_1}{\alpha_1 - \beta^i_1}}^2 = \abs{\frac{\alpha_1 - \beta^{i+1}_1}{\alpha_1 - \beta^i_1}}^2 .\label{eqn:intstep}
\end{align}
According to \eqref{eqn:intstep} we get
\begin{align}
     \mu(\beta^i_1) &\leq \mu(\beta^{i+1}_1) \text{ for }  i \in \{1,2,\dots,k_1^*-1\},\nonumber\\
    \mu(\beta^i_1) &\geq \mu(\beta^{i+1}_1) \text{ for }  i \in \{k_1^*,k_1^*+1,\dots,K_1-1\}. \qedhere 
\end{align}
\end{proof}

\subsection{Proof of Proposition \ref{cor:reducedmeanRSS}}
\label{subsec:ProofCorollary1}
\begin{proof}
    By \eqref{eqn:expectedRSS} and \eqref{eqn:contphaseHMMIMOuserchannel} for $\alpha_1$ and $\alpha_2$, we have the following
\begin{align}
    &\mu(\alpha_1,\alpha_2) = G + \sigma^2. \label{eqn:stepalpha1}
\end{align}
By \eqref{eqn:expectedRSS} and \eqref{eqn:contphaseHMMIMOuserchannel}, for $\beta^{k_1^*}_1 = \alpha_1 \pm \epsilon_1$, and $\beta^{k_2^*}_2 = \alpha_2 \pm \epsilon_2$,  we get
\begin{align}
    \mu(\beta^{k_1^*}_1, \beta^{k_2^*}_2)  
    &= G\big|\mathrm{sinc}\big(\pm K_x\pi \epsilon_1\big)\mathrm{sinc}\big(\pm K_y\pi\epsilon_2\big)\big|^2 + \sigma^2\label{eqn:stepbetaopt}.
\end{align}
As $\mathrm{sinc}(-x)=\mathrm{sinc}(x)$, the difference between \eqref{eqn:stepalpha1} and \eqref{eqn:stepbetaopt} is given by
\begin{align}
\label{eqn:epsilon12meanRSS}
& \big|\mu(\alpha_1,\alpha_2) -  \mu(\beta^{k_1^*}_1, \beta^{k_2^*}_2)\big|\nonumber\\
& = G\big[1 - \big|\mathrm{sinc}\big(K_x\pi\epsilon_1\big)\mathrm{sinc}\big(K_y\pi \epsilon_2\big)\big|^2\big].
\end{align}
The worst case can happen when $\mu_1(\beta^k_1)$ is the same for the two discrete values in the central lobe. This can be possible as the two points will lie on opposite sides of $\alpha_1$. Therefore, 
this worst case scenario would happen when $\epsilon_1 = \frac{1}{2K_x}$. Thereby, if we consider any $\epsilon_1 \in \big[0, \frac{1}{2K_x}\big)$, then 
\begin{align}
\label{eqn:Ineqalpha1}
\big|\mathrm{sinc}\big(K_x\pi(\alpha_1 - \beta^{k_1^*}_1\big)\big| 
\geq \big|\mathrm{sinc}\big(\pi/2\big)\big| = 2/\pi.
\end{align}
Similarly, the worst case scenario of $\mu_2(\beta^k_2)$ would happen when $\epsilon_2 = \frac{1}{2K_y}$, and, thereby for any $\epsilon_2 \in \big[0, \frac{1}{2K_y}\big)$, we get
\begin{align}
\label{eqn:Ineqalpha2}
\big|\mathrm{sinc}\big(K_y\pi(\alpha_2 - \beta^{k_2^*}_2\big)\big| = \big|\mathrm{sinc}\big(K_y\pi\epsilon_2\big)\big| \geq 2/\pi.
\end{align}
Hence, by applying \eqref{eqn:Ineqalpha1} an \eqref{eqn:Ineqalpha2} in \eqref{eqn:epsilon12meanRSS}, for $0 \leq \epsilon_1 < \frac{1}{2K_x}$ and $0 \leq \epsilon_2 < \frac{1}{2K_y}$, we get the bound in \eqref{eqn:epsilondiffbound}.
\end{proof}


\subsection{Proof of Theorem \ref{thm:PSHMTBound}}
\label{subsec:ProofThm1}
\vspace{-1mm}
\begin{proof}
{\it \HB} divides $n$ pilot symbols into $n/2$ for each phase $i$ which is further splitted into $L_i+1$ batches that satisfies \eqref{eqn:L_T}, where $L_i=\frac{\log_2 K_i/3}{\log_2 3/2}$ and outputs the phase-shift parameter $\beta^{\hat{k}_i}_i$ for phases $i=1,2$. Let us denote
\begin{align}
    \hspace{-2mm}I:= \Pro \big (\beta^{\hat{k}_1}_1 \neq \beta^{k_1^*}_1\big) \text{ and } II:= \Pro \big(\beta^{\hat{k}_2}_2 \neq \beta^{k_2^*}_2\big).
    \label{eqn:errProbbeta2}
\end{align}
We will now upper bound the probability of error as,
\begin{align}
\label{eqn:UnionBound}
   \Pr \bigg (\bigcup\limits_{i=1,2} \Big\{\beta^{\hat{k}_i^\pi}_i \neq \beta^*_i\Big\}\bigg) = I + II
\end{align}
\begin{itemize}
    \item \textbf{Upper Bound of I :} 
    As $\beta^{\hat{k}_1}_1$ can be eliminated in any phase $l = 1,\dots, L_1+1$, we get   
\begin{align}
\label{eq:ubound}
    \Pro(\beta^{\hat{k}_1}_1 \neq \beta^{k_1^*}_1)
    &\le \sum_{l=1}^{L_1+1}\Pr(\beta^{k_1^*}_1 \text{ elim. in } l).
\end{align}
where $L_1 = \frac{\log_2K_1/3}{\log_2 3/2}$. In phase $l$, the following cases can happen:

\begin{enumerate}
    \item \textbf{Case 0:} $\beta^{k_1^*}_1 \in \{\beta^{k^B}_1,\dots, \beta^{k^C}_1\}$, and $\hat{\mu}^l_{\beta^{k^C}_1}$ or $\hat{\mu}^l_{\beta^{k^D}_1}$ is greater than both $\hat{\mu}^l_{\beta^{k^A}_1}$ and $\hat{\mu}^l_{\beta^{k^B}_1}$, or $\hat{\mu}^l_{\beta^{k^A}_1}$ or $\hat{\mu}^l_{\beta^{k^B}_1}$ is greater than both $\hat{\mu}^l_{\beta^{k^C}_1}$ and $\hat{\mu}^l_{\beta^{k^D}_1}$.
    \item \textbf{Case 1:} $\beta^{k_1^*}_1 \in \{\beta^{k^A}_1,\dots, \beta^{k^B}_1\}$, and $\hat{\mu}^l_{\beta^{k^C}_1}$ or $\hat{\mu}^l_{\beta^{k^D}_1}$ is greater than both $\hat{\mu}^l_{\beta^{k^A}_1}$ and $\hat{\mu}^l_{\beta^{k^B}_1}$.
    \item \textbf{Case 2:} $\beta^{k_1^*}_1 \in \{\beta^{k^C}_1,\dots, \beta^{k^D}_1\}$, and $\hat{\mu}^l_{\beta^{k^A}_1}$ or $\hat{\mu}^l_{\beta^{k^B}_1}$ is greater than both $\hat{\mu}^l_{\beta^{k^C}_1}$ and $\hat{\mu}^l_{\beta^{k^D}_1}$.
\end{enumerate}
Note that $\beta^{k_1^*}_1$ will get eliminated if $\beta^{k_1^*}_1$ falls either in Case 1 or Case 2 but will not get eliminated if it falls in Case 0. Therefore, Case 0 is not favourable here. Further note that Case 1 and Case 2 are symmetrical, illustrated in Fig. \ref{case12}. Therefore, to get the upper bound of the error probability, we will focus on Cases 1 and 2, and without loss of generality, we  consider Case 1. 
\begin{figure}[h]
    \centering
    \begin{tikzpicture}[line width=1.5pt]
        \draw (-0.2,0) --++(7,0) ; 
        \foreach \i in{0,2.2,4.4,6.6}{\draw [thin](\i,0.1)--++(0,-0.2);}
        \draw[red,dashed] (0, 0.2) .. controls (5.3,1.8) .. (6.6, 0.4);
        \draw[olive,densely dashdotted] (0, 0.3) .. controls (3.3,1.8) .. (6.6, 0.3);
        \draw[blue] (0, 0.4) .. controls (1.3,1.8) .. (6.6, 0.2);
        \draw[blue, dotted, line width=1] (1.5,1.4)--++(0,-1.3)node[below]{$\beta^{k_1^*}_1$};
        \draw[red, dotted, line width=1] (5.2,1.4)--++(0,-1.3)node[below]{$\beta^{k_1^*}_1$}; 
        \draw[olive, densely dashdotted, line width=1] (3.5,1.4)--++(0,-1.3)node[below]{$\beta^{k_1^*}_1$};
        \node at (0,-0.3) {{$\beta^{k^A}_1$}};
        \node at (2.2,-0.3) {{$\beta^{k^B}_1$}};
        \draw [decorate,decoration={brace,mirror,amplitude=10pt},yshift=0pt,line width=1pt] (2.2,-0.5) -- (4.4,-0.5) ;
        \node at (3.3,-1.1) {\footnotesize$\Delta^1_{B_1,C_1}$};
        \node at (4.4,-0.3) {{$\beta^{k^C}_1$}};
        \node at (6.6,-0.3) {{$\beta^{k^D}_1$}};
        \node[blue] at (0.2,1.7){Case 1};
        \node[red] at (6.2,1.7){Case 2};
        \node[olive] at (3.2,1.7){Case 0};
    \end{tikzpicture}
    \caption{Different cases of elimination in any phase $l$.} 
    \label{case12}
    \vspace{-5mm}
\end{figure}
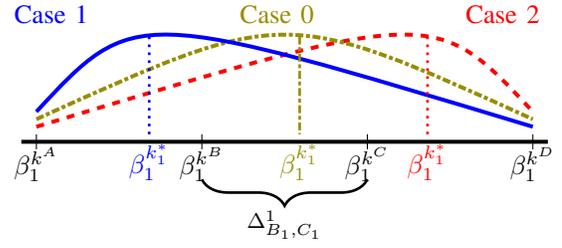
\begin{align}
  &\Pro\big(\beta^{k_1^*}_1 \text{ elim. in } l\big) \nonumber\\
  &\le \Pro\bigg(\hat{\mu}^l_{\beta^{k^C}_1} > \hat{\mu}^l_{\beta^{k^A}_1} \text{and}\,\, \hat{\mu}^l_{\beta^{k^B}_1}| \beta^{k_1^*}_1\in\{\beta^{k^A}_1,\dots,\beta^{k^B}_1\}\bigg) \nonumber\\
  &+ \Pro\bigg(\hat{\mu}^l_{\beta^{k^D}_1} > \hat{\mu}^l_{\beta^{k^A}_1} \text{and}\,\, \hat{\mu}^l_{\beta^{k^B}_1}| \beta^{k_1^*}_1\in\{\beta^{k^A}_1,\dots,\beta^{k^B}_1\}\bigg) \label{eq:mu_b.1}
  \end{align}
  For Case 1, $\mu_1(\beta^{k^C}_1)\ge\mu_1(\beta^{k^D}_1)$ by unimodality. Therefore, we can upper bound \eqref{eq:mu_b.1} as
  \begin{align}
  \label{eq:mu_b}
   &\Pro\big(\beta^{k_1^*}_1 \text{ elim. in } l\big) \nonumber\\
          &\le 2\Pro\bigg(\hat{\mu}^l_{\beta^{k^C}_1}> \hat{\mu}^l_{\beta^{k^A}_1} \text{and}\,\, \hat{\mu}^l_{\beta^{k^B}_1}| \beta^{k_1^*}_1\in\{\beta^{k^A}_1,..,\beta^{k^B}_1\}\bigg),
\end{align}
Now for Case 1, $\mu_1(\beta^{k^B}_1)$ is always greater than $\mu_1(\beta^{k^C}_1)$, but $\mu_1(\beta^{k^A}_1)$ may not be greater than  $\mu_1(\beta^{k^C}_1)$. Then, we can further upper bound \eqref{eq:mu_b} as
\begin{align}
\label{eqn:prob_ub}
     &\Pro\big(\beta^{k_1^*}_1 \text{ elim. in } l\big)\nonumber\\
     &\le 2\Pro\bigg(\hat{\mu}^l_{\beta^{k^C}_1} > \hat{\mu}^l_{\beta^{k^B}_1}|\beta^{k_1^*}_1\in\{\beta^{k^A}_1,\dots,\beta^{k^B}_1\}\bigg).
\end{align}

We next upper bound the right-hand side of \eqref{eqn:prob_ub}. We denote $\chi_{\nu}^2(q)$ as a non-central Chi-squared distribution with $\nu$ degrees of freedom and non-centrality parameter $q$. As $r_1(\beta^k_1)$ follows non-central Chi-squared distribution, according to \cite[Lemma 2]{Arxiv2022learningHMT}, we denote
\begin{align}
    \hat{\mu}_{B_1} &:=\frac{2m_1}{\sigma^2}\hat{\mu}_{\beta^{k^B}_1} \sim \chi^2_{2m_1}(m_1q_{B_1}) \label{eqn:empmeansbeta1nextbeta2}\\
     \hat{\mu}_{C_1} &:=\frac{2m_1}{\sigma^2}\hat{\mu}_{\beta^{k^C}_1} \sim \chi^2_{2m_1}(m_1q_{C_1}) \label{eqn:empmeanbeta1beta2}\\
     \mu_{B_1} &:= \mu_1(\beta^{k^B}_1) = \sigma^2 + \frac{\sigma^2}{2}q_{B_1}\label{eqn:meanbeta1nextbeat2}\\
    \mu_{C_1} &:= \mu_1(\beta^{k^C}_1) = \sigma^2 + \frac{\sigma^2}{2}q_{C_1}\label{eqn:meanbeta1beat2}
\end{align}
where $m_1 = \frac{N^l_1}{4}$, $q_{B_1} = \frac{2\big|\sqrt{P}H^{ff}(\beta^{k^B}_1, \beta^0_2)\big|^2}{\sigma^2}$ and $q_{C_1} = \frac{2\big|\sqrt{P}H^{ff}(\beta^{k^C}_1, \beta^0_2)\big|^2}{\sigma^2}$.
Now we will use the upper and lower tail bounds for non-central Chi-squared distribution \cite{JournalStats2021exponential} to derive the upper bound in \eqref{eqn:prob_ub}. 
\begin{lem}
\label{lem:Exponentialbound}
Let $X \sim \chi_{\nu}^2(q_1)$, $q_1 > 0$ and  $Y \sim \chi_{\nu}^2(q_2)$, $q_2 > 0$ and $q_1 < q_2$. $X$ and $Y$ are independent. Then,
\begin{align}
\label{eqn:Exponentialbound}
    \Pro(X > Y) \leq 2\exp\bigg\{ \frac{-\nu (\nu + q_1)^2(q_2 - q_1)^2}{4(\nu + 2q_2)^2(2\nu + q_1 + q_2)^2} \bigg\}
\end{align}
\end{lem}
\vspace{-3mm}
\begin{proof}
    We will use Theorem 3 and 4 of \cite{JournalStats2021exponential} to prove the above lemma. Refer \cite{JournalStats2021exponential} for the proof of these theorems. 
    \begin{align}
        \Pro(X > Y) & = \Pro\bigg(\frac{X/\nu}{Y/\nu} > 1\bigg)\label{eqn:DfF}
    \end{align}
    We denote $F = \frac{X/\nu}{Y/\nu}$ as doubly non-central F distrbution. According to \cite{JournalStats2021exponential}, we can write \eqref{eqn:DfF} as
    \begin{align}
      \Pro\big(F > 1\big)  \leq &\Pro\big(X > (\nu + q_1)(1 + \delta) \big) \nonumber\\
      &+ \Pro\big(Y < (\nu + q_2)(1 - \delta)\big) \label{eqn:XandYbound}
    \end{align}
    where $\delta = \frac{q_2 - q_1}{2\nu + q_1 + q_2}$ and $0 < \delta < 1$. By \cite[Thm. 3]{JournalStats2021exponential}, we obtain
    \begin{align}
        \label{eqn:Xbound}
        \Pro&\big(X > (\nu + q_1)(1 + \delta) \big) \nonumber\\
        &\leq \exp\bigg\{\frac{-\nu(\nu + q_1)^2\delta^2}{4(\nu + 2q_1)(\nu + 2q_1 + (\nu + q_1)\delta)} \bigg\}
    \end{align}
    By \cite[Thm. 4]{JournalStats2021exponential} and $q_1 < q_2$, we obtain
    \begin{align}
        \label{eqn:Yboundrevised}
        &\Pro\big(Y < (\nu + q_2)(1 - \delta) \big) \leq \exp\bigg\{\frac{-\nu(\nu + q_1)^2\delta^2}{4(\nu + 2q_2)^2} \bigg\}.    \end{align}
    \vspace{-2mm}
Applying $\nu, q_1, q_2 > 0$ and $q_1 < q_2$, we obtain
\begin{align}
     &\big(\nu + 2q_2\big)^2 > (\nu + 2q_1)\big(\nu + 2q_1 + (\nu + q_1)\delta\big).\label{eqn:ii}
\end{align}
    Applying \eqref{eqn:ii} in \eqref{eqn:Xbound}, we obtain
    \begin{align}
    \label{eqn:reXboundrevised}
         \Pro\big(X > (\nu + q_1)(1 + \delta) \big) \leq \exp\bigg\{\frac{-\nu(\nu + q_1)^2\delta^2}{4(\nu + 2q_2)^2} \bigg\}.
    \end{align}
Applying \eqref{eqn:reXboundrevised} and \eqref{eqn:Yboundrevised} in \eqref{eqn:XandYbound} we get \eqref{eqn:Exponentialbound}.
\end{proof}
For our case $\nu = \frac{N^l_1}{2}, q_1 = \frac{N^l_1}{4}q_{C_1}, q_2 = \frac{N^l_1}{4}q_{B_1}$, and $\Delta^l_{B_1,C_1}= \mu_{B_1}- \mu_{C_1} = \frac{\sigma^2}{2}(q_{B_1} - q_{C_1})$. Applying these facts and Lemma \ref{lem:Exponentialbound} in \eqref{eqn:prob_ub}, we obtain
\begin{align}
         &\Pro\big(\beta^{k_1^*}_1 \text{ elim. in } l\big)\nonumber\\
         &\le 4\exp\bigg\{\frac{-N^l_1(\Delta^l_{B_1,C_1})^2(\mu_{C_1})^2}{8 (2\mu_{B_1} - \sigma^2)^2(\mu_{B_1} + \mu_{C_1})^2}\bigg\}.\label{eqn:bound2}
\end{align}
Using the value of $G$, 
$\forall k \in \{1,2,\dots,K_i\}$ for $i=1,2$ the mean RSS can be bounded as
\begin{align}
\label{eqn:meanRSSbound}
    \sigma^2 \leq \mu_i(\beta^k_i) \leq G + \sigma^2.
\end{align}
Applying \eqref{eqn:meanRSSbound} in \eqref{eqn:bound2}, we obtain
\begin{align}
     & \Pro\big(\beta^{k_1^*}_1 \text{ elim. in } l\big)\leq 4\exp\bigg\{-\frac{N^l_1(\Delta^l_{B_1,C_1})^2}{32\sigma^4(1+3G/\sigma^2)^4}\bigg\}.\label{eqn:mubrevised}
\end{align}
For case 1, $\Delta^l_{B_1,C_1} > 0$ for the unimodal structure between $\mu_{B_1}$ and $\mu_{C_1}$. We define $\Delta_1 = \min\limits_{2\leq k \leq K_1-1}\big|\mu_1(\beta^k_1) - \mu_1(\beta^{k-1}_1)\big|$. By the fact that there are at least $j_l/3$ parameters between $\beta^{k^B}_1$ and $\beta^{k^C}_1$, we have $\Delta^l_{B_1,C_1} \geq (j_l/3) \Delta_1$.
Thus from \eqref{eqn:mubrevised} we get
\vspace{-2mm}
\begin{align}
\label{eq:tlkl}
   & \Pro(\beta^{k_1^*}_1 \text{ elim. in } l)\le  4\exp\bigg\{\frac{-N^l_1\big(j_l\Delta_1\big)^2}{288\sigma^4(1+3G/\sigma^2)^4}\bigg\}.
\end{align}
Using $j_l=\left(2/3\right)^{l}K$ in \eqref{eq:tlkl}, we can find the probability of the best value is eliminated in batches 1, 2, $L_1+1$, and the rest of the batches separately. Using \eqref{eqn:Nl} and \eqref{eqn:L} for batch 1 and 2, we have
\vspace{-5mm}
 \begin{align}
 \label{eq:round12}
    \Pro(&\beta^{k_1^*}_1 \text{ elim. in } 1\& 2)
    \le 4\exp{\left\{-\frac{nK_1\Delta_1^2}{576\sigma^4(1+3G/\sigma^2)^4}\right\}} \nonumber\\
    &+ 4\exp\left\{-\frac{nK_1\Delta_1^2}{1296\sigma^4(1+3G/\sigma^2)^4}\right\}.\nonumber\\
    &\leq 8\exp\left\{-\frac{nK_1\Delta_1^2}{1296\sigma^4(1+3G/\sigma^2)^4}\right\}.
\end{align}
For $L_1+1$ where the best value is selected among three values 
and each value is sampled $n/18$ times,  we get
\begin{align}
\label{eq:L1}
    &\Pro(\beta^{k_1^*}_1\text{ elim. in batch  }L_1+1)\nonumber\\
    &\le 4\exp\left\{-\frac{n\Delta_1^2}{1296\sigma^4(1+3G/\sigma^2)^4}\right\} 
\end{align}
By \eqref{eq:tlkl}, the error probability for remaining phases is
\begin{align}\label{eq:Lrest}
     &\Pro(\beta^{k_1^*}_1 \text{ elim. in batch 3 to batch $L_1$})\nonumber\\
     &\le 4\sum_{l=3}^{L_1}\exp\left\{\frac{-n(K_1)^2\Delta_1^2}{576\sigma^4(1+3G/\sigma^2)^4}\frac{2^{L_1-l+1}}{3^{L_1-l+2}}\left(\frac{2}{3}\right)^{2l}\right\}\nonumber\\
     &\le 4(L_1-2)\exp\left\{-\frac{n \Delta_1^2}{1296\sigma^4(1+3G/\sigma^2)^4}\right\}. 
\end{align}
By 
\eqref{eq:L1} and \eqref{eq:Lrest}, the upper bound is
\begin{align}
    &\Pro(\beta^{\hat{k}_1}_1 \neq \beta^{k_1^*}_1)\le 8\exp\left\{-\frac{nK_1\Delta_1^2}{1296\sigma^4(1+3G/\sigma^2)^4}\right\} \nonumber\\
       &+4(L_1-1)\exp\left\{-\frac{n \Delta_1^2}{1296\sigma^4(1+3G/\sigma^2)^4}\right\}.\label{eqn:Ibound}
\end{align}
\vspace{-4mm}
    \item \textbf{Upper Bound of II :} We have $L_2= \frac{\log_2 K_2/3}{\log_2 3/2}$.  Following the same lines of proof from \eqref{eq:ubound}-\eqref{eqn:prob_ub}, we obtain
    \begin{align}
\label{eq:ubound2}
    &\Pro(\beta^{\hat{k}_2}_2 \neq \beta^{k_2^*}_2) \le \sum_{l=1}^{L_2+1}\Pr(\beta^{k_2^*}_2 \text{ elim. in } l).\nonumber\\
     &\le 2\sum_{l=1}^{L_2+1}\Pro\bigg(\hat{\mu}^l_{\beta^{k^C}_2} > \hat{\mu}^l_{\beta^{k^B}_2}|\beta^{k_2^*}_2\in\{\beta^{k^A}_2,\dots,\beta^{k^B}_2\}\bigg).
\end{align}
We next upper bound the right-hand side of \eqref{eq:ubound2}. According to \cite[Lemma 2]{Arxiv2022learningHMT},
\begin{align}
    \hat{\mu}_{B_2} &:=\frac{2m_2}{\sigma^2}\hat{\mu}_{\beta^{k^B}_2} \sim \chi^2_{2m_2}(m_2q_{B_2}) \label{eqn:2empmeansbeta1nextbeta2}\\
     \hat{\mu}_{C_2} &:=\frac{2m_2}{\sigma^2}\hat{\mu}_{\beta^{k^C}_2} \sim \chi^2_{2m_2}(m_2q_{C_2}) \label{eqn:2empmeanbeta1beta2}\\
     \mu_{B_2} &:= \mu_2(\beta^{k^B}_2) = \sigma^2 + \frac{\sigma^2}{2}q_{B_2}\label{eqn:2meanbeta1nextbeat2}\\
    \mu_{C_2} &:= \mu_2(\beta^{k^C}_2) = \sigma^2 + \frac{\sigma^2}{2}q_{C_2}\label{eqn:2meanbeta1beat2}
\end{align}
where  $m_2 = \frac{N^l_2}{4}$, $q_{B_2} = \frac{2\big|\sqrt{P}H^{ff}(\beta^{\hat{k}_{L_1+1}}_1, \beta^{k^B}_2)\big|^2}{\sigma^2}$ and $q_{C_2} = \frac{2\big|\sqrt{P}H^{ff}(\beta^{\hat{k}_{L_1+1}}_1, \beta^{k^C}_2)\big|^2}{\sigma^2}$.
Furthermore, we denote 
\begin{align*}
    &\Delta^l_{B_2,C_2} = \mu_{B_2}- \mu_{C_2} > 0,\\
    &\Delta_2 = \min\limits_{2\leq k \leq K_2-1}\big|\mu_2(\beta^k_2) - \mu_2(\beta^{k-1}_2)\big|,
\end{align*}
where $\Delta^l_{B_2,C_2} \geq (j_l/3) \Delta_2$. Using these facts in the  following lines of proof from \eqref{eqn:Exponentialbound}-\eqref{eq:Lrest}, we obtain
\begin{align}
   & \Pro(\beta^{\hat{k}_2}_2 \neq \beta^{k_2^*}_2)  \le 8\exp\left\{-\frac{nK_2\Delta_2^2}{1296\sigma^4(1+3G/\sigma^2)^4}\right\}\nonumber\\
       &+ 4(L_2-1)\exp\left\{-\frac{n\Delta_2^2}{1296\sigma^4(1+3G/\sigma^2)^4}\right\}.\label{eqn:IIbound}
\end{align}
\end{itemize}
Combining \eqref{eqn:Ibound} and \eqref{eqn:IIbound}, we obtain the bound in \eqref{eqn:PSHMTBound}.\qedhere
\end{proof}

\bibliographystyle{IEEEtran}
\bibliography{ref}

\end{document}